\documentclass[11pt]{article}

\usepackage[margin=1in]{geometry}
\usepackage{amsmath,amssymb,amsthm,mathtools}
\usepackage{booktabs}
\usepackage{graphicx}

\usepackage{microtype}
\usepackage[dvipsnames]{xcolor}
\usepackage[colorlinks=true,citecolor=blue,linkcolor=blue,urlcolor=blue]{hyperref}

\newtheorem{theorem}{Theorem}
\newtheorem{lemma}{Lemma}
\newtheorem{proposition}{Proposition}

\newtheorem{observation}{Observation}

\theoremstyle{definition}

\newtheorem{algorithm}{Algorithm}
\theoremstyle{remark}

\newcommand{\Greedy}{\operatorname{Greedy}}
\newcommand{\cl}{\operatorname{cl}}
\newcommand{\E}{\mathbb{E}}
\newcommand{\Prb}{\mathbb{P}}

\title{A $3.1462$-Competitive Algorithm for Matroid Secretary}
\author{Jianan Lin$^{1}$\quad Chenhao Wang$^{2,3}$\\[0.75em]
$^{1}$Rensselaer Polytechnic Institute\\
$^{2}$Beijing Normal University--Zhuhai\\
$^{3}$Beijing Normal--Hong Kong Baptist University}
\date{}

\begin{document}
\maketitle

\begin{abstract}
In the matroid secretary problem, a set of elements arrives one at a time in a uniformly random order, and the weight of each element is revealed only when it arrives.
The goal is to  maximize the total weight of the selected elements, subject to the constraint that they form an independent set of a given matroid.
An online algorithm must decide immediately and irrevocably whether to select each arriving element.
Singla (2026) recently resolved the long-standing open question of whether a constant competitive ratio is achievable for every matroid by giving a $4$-competitive algorithm.
His algorithm requires only the number of elements and the ability to ask, for any set of already-arrived elements, whether that set is independent in the matroid.

With access to the same information, we provide an algorithm, Dynamic Thinning, with an improved competitive ratio of $C_*\approx3.1462$, where $C_*-\log C_*=2$.
Like Singla's algorithm, it maintains a set of previously seen elements and accepts an arriving element only if it belongs to the maximum-weight independent subset of the maintained set together with that element.
In addition, between consecutive arrivals, the algorithm may randomly remove unselected elements from the maintained set, and in certain cases it accepts an arriving element with a probability that depends on when it arrives.
The frequency of these removals and the acceptance probabilities are chosen so that each element of a fixed maximum-weight independent set is selected with probability exactly $\frac{1}{C_*}$, which immediately yields the competitive ratio.
\footnote{After this paper, Huang \cite{huang2026settling} settled the problem completely by giving an $e$-competitive algorithm with access to the same information as Singla's algorithm.
Joyce \cite{joyce2026multilevel} later extended Dynamic Thinning by maintaining several nested sets, obtaining an $(e+\varepsilon)$-competitive algorithm for every fixed $\varepsilon>0$. 
}
\end{abstract}

\section{Introduction}
\label{sec:introduction}

In many real-world scenarios, decisions are often made subject to various constraints before all possible options are known. 
For example, an organization interviews candidates one at a time to fill several positions and must make a hiring decision for each candidate before interviewing the next \cite{kleinberg2005multiple}.
Similarly, an online advertiser with a budget for a limited number of ad impressions observes website visitors one at a time and must decide whether to spend an impression on each visitor before the next visitor arrives \cite{babaioff2007knapsack}. 
To capture these and many other scenarios, Babaioff, Immorlica, and Kleinberg \cite{babaioff2018matroid} introduced the matroid secretary problem nearly two decades ago, where decisions must be made over time subject to matroid constraints.  
The problem and its variants have also been applied to online mechanism design \cite{hajiaghayi2004adaptive, babaioff2008online}, online matching \cite{korula2009algorithms, kesselheim2013optimal}, interval scheduling \cite{im2011secretary}, and network design \cite{babaioff2018matroid, korula2009algorithms}.

In the matroid secretary problem, a set of $n$ elements (e.g., candidates or visitors) arrives one at a time in a uniformly random order. 
Each element has a nonnegative weight, which is revealed only when the element arrives. 
An online algorithm must decide immediately and irrevocably whether to select each arriving element before seeing the next element.
The goal is to  maximize the total weight of the selected elements, subject to the constraint that they form an independent set of a given matroid. 
Examples of specific matroids include the uniform matroid, whose independent sets are all sets of at most $k$ elements, and the graphic matroid, whose elements are the edges of a graph and whose independent sets are the sets of edges containing no cycle. 

Babaioff, Immorlica, and Kleinberg \cite{babaioff2018matroid} provided an online algorithm with an $O(\log r)$ competitive ratio for general matroids, where $r$ is the matroid rank (i.e., the size of its largest independent set), and gave online algorithms that achieve a constant competitive ratio for several classes of matroids. 
They conjectured that there is an online algorithm that achieves a constant competitive ratio for every matroid. 
For general matroids, the competitive ratio was later improved to $O(\sqrt{\log r})$ \cite{chakraborty2012improved} and then to $O(\log\log r)$ \cite{lachish2014competitive,feldman2018simple}. 
For many special classes of matroids, including graphic, transversal, laminar, and regular matroids, online algorithms with constant competitive ratios were already known \cite{korula2009algorithms,im2011secretary,dinitz2014regular}.

\paragraph{Recent Breakthroughs.}
Recently, Singla~\cite{singla2026matroid} resolved this conjecture by giving a $4$-competitive algorithm for every matroid. 
Moreover, Singla's algorithm selects each element of a fixed optimum (i.e., a maximum-weight independent set in the matroid that is fixed by a tie-breaking rule) with probability at least $\frac{1}{4}$.  
An online algorithm with such a per-element guarantee (i.e., one that selects each element of a fixed optimum with probability at least $\frac{1}{c}$) is called $c$-probability-competitive~\cite{soto2021strong}, and it is also $c$-competitive. 

Without being given the matroid as input, Singla's algorithm requires only the number of elements in advance and the ability to ask, for any set of already-arrived elements (selected or not), whether that set is independent in the matroid. 
The algorithm builds on a variant of the secretary problem (i.e., the two-sided game of Googol~\cite{correa2022twosided}) and maintains a reference set of previously seen elements that is updated as elements are processed. 
Independently, Abdi et al.~\cite{abdi2026strong} gave a $64$-competitive algorithm under the stronger assumption that the entire matroid (i.e., which sets are independent) is known to the algorithm in advance. 
Abdi et al.~\cite{abdi2026strong} and Bérczi et al.~\cite{berczi2026strongsecretaryconjecturetrue} also gave $e$-competitive algorithms for linear matroids, which include graphic, regular, and laminar matroids as special cases. 

\subsection{Our Contribution and Techniques}
In this paper, we provide an online algorithm with a competitive ratio of $C_*\approx3.1462$ that has access to the same information as Singla's algorithm. 
The constant $C_*$ is the unique number greater than $1$ satisfying $C_*-\log C_*=2$. 
Moreover, our algorithm (Theorem~\ref{thm:main}) selects each element of a fixed optimum with probability exactly $\frac{1}{C_*}\approx0.3178$, asks whether a set is independent $O(n^2)$ times, and is ordinal (i.e., the algorithm only needs to know which of two elements has the larger weight regardless of the magnitudes).

Both Singla's algorithm and our algorithm maintain a set of previously seen elements and repeatedly compute a maximum-weight independent subset to guide their online decisions. Such an idea has also appeared in Kesselheim et al. \cite{kesselheim2013optimal}. 
For a matroid, a maximum-weight independent set can be computed by the greedy rule that scans the elements in decreasing order of weight and adds each element that keeps that set independent. 
For a set $Y$ of elements, we write $\Greedy(Y)$ to denote the maximum-weight independent subset of $Y$ computed by the greedy rule of a matroid. 
Both algorithms maintain a \emph{reference set} $X$ of previously seen elements and $\Greedy(X)$, which we call the \emph{reference basis}. 
An arriving element $e$ is a candidate for selection only if it belongs to $\Greedy(X\cup\{e\})$. 

\subsubsection*{Singla's Algorithm}

We first restate an equivalent description of Singla's algorithm.
First, draw $K$ from the binomial distribution $\operatorname{Bin}\big(n,\frac{1}{2}\big)$, where $n$ is the number of elements, and observe and reject the first $K$ arriving elements. Call this set the sample $S$ and initialize the reference set to $X=S$. 

After rejecting the sample, the algorithm considers the later elements in their arrival order and, in between, revisits the sample elements one at a time to consider removing them from $X$. 
The algorithm achieves this by constructing a virtual order of the $n$ elements. 
To construct the virtual order, the algorithm picks $K$ of the $n$ positions uniformly at random and places the sample elements in these positions in their arrival order.
The later arriving elements fill the remaining positions, also in their arrival order.
This can be done online because $n$ and $K$ are both known after rejecting the sample. 

Let $F$ denote the set of elements processed so far in the virtual order and initialize $F=\varnothing$. 
When processing an element $e$ in the virtual order, the algorithm proposes to update the reference set to $X'=X\setminus\{e\}$ if $e\in S$, and $X'=X\cup\{e\}$ otherwise. 
The update $X\leftarrow X'$ is carried out if and only if 
\[
\Greedy(X')\cap F=\Greedy(X)\cap F,
\]
that is, if and only if the update does not change which already-processed elements belong to the reference basis. 
A new arriving element $e$ not in the sample is selected when both conditions hold: it is a candidate for selection, i.e., $e\in\Greedy(X\cup\{e\})$, and the update is carried out. 
Whether or not the update is carried out, $e$ is then added to $F$.

\subsubsection*{Our Algorithm and Analysis}
In Singla's algorithm, the update condition ensures that, once an element has been processed and added to $F$, its membership in the reference basis, as determined at the end of that step, never changes.
The reference basis can still change, but only for elements that have not yet been processed. 
For convenience, we refer to an element of $F$ as protected. 
Note that the reference basis consists of two kinds of elements, selected arriving elements and rejected sample elements, and Singla's update condition protects both kinds. 

On the other hand, our algorithm protects only the selected elements. 
In particular, the algorithm maintains $A\subseteq\Greedy(X)$ after every update, where $A$ is the set of selected elements (called the accepted set in Section \ref{sec:algorithm}). 
An element of $\Greedy(X)\setminus A$ has not been selected and may be replaced by a selected arriving element or deleted from $X$.
Our algorithm uses a selection rule to decide whether to accept a candidate arriving element and a deletion rule to determine when to delete an unselected element from the reference basis.

We choose these rules so that the probability of every possible reference set $X$ and selected set $A$ is given by an explicit formula (Proposition~\ref{prop:snapshot}).
This formula lets us compute the probability that each element of the fixed optimum is selected and choose the parameter of the algorithm to minimize the competitive ratio.
To design rules that admit such a formula, we build on an observation that Abdi et al.~\cite{abdi2026strong} made for a similar algorithm for the single-sample prophet problem. 
Their observation is that, at the end of the algorithm, given the reference set $X$, the selected set $A$ is a uniformly random subset of $\Greedy(X)$. 
That is, each element of $\Greedy(X)$ belongs to $A$ independently with probability $\frac{1}{2}$. 
We adapt this observation to the matroid secretary problem by replacing the fixed $\frac{1}{2}$ with a probability that changes as more elements arrive.

Concretely, we generate $n$ independent uniform random numbers in $[0,1]$ and assign them, in increasing order, to the arriving elements. 
These numbers give each element $e$ an arrival time $\tau_e$. 
For a parameter $p\in(0,1)$, we let the sample $S$ be the set of elements with $\tau_e\le p$ and set $X=S$ and $A=\varnothing$.
Given the arrival times, our deletion rule works as follows. 
Each element of $\Greedy(X)\setminus A$ (i.e., each unselected element of the reference basis) is given its own timer whose length is drawn independently from the exponential distribution with mean $p$. 
When the timer of an element runs out, the element is deleted from $X$, which is equivalent to deleting each such element at rate $\frac{1}{p}$ (see Section~\ref{sec:algorithm}). 

For our selection rule, when an element $e$ arrives at time $t=\tau_e>p$, 
if $e\notin \Greedy(X\cup\{e\})$, then $e$ is not a candidate for selection and is rejected. 
Otherwise, there are two cases (Lemma~\ref{lem:greedy-exchange}).
If $\Greedy(X\cup\{e\})=(\Greedy(X)\setminus\{f\})\cup\{e\}$ for some $f\in \Greedy(X)$ (i.e., adding $e$ replaces exactly one element $f$ of the reference basis), then $e$ is selected if $f\notin A$ and rejected otherwise to ensure that a selected element is never replaced. 
If $\Greedy(X\cup\{e\})=\Greedy(X)\cup\{e\}$ (i.e., the reference basis is expanded with $e$), then $e$ is selected with probability $q(t)=e^{-(t-p)/p},$
which equals $1$ at $t=p$ and decreases afterward. Notice that only selected elements are added to $X$.

\paragraph{Analysis.}
We set the timer mean to $p$ and the selection probability to $q(t)$ to ensure two properties at every time $t\in[p,1]$. 
Let $X_t$ and $A_t$ denote the reference set and the selected set at time $t$. 
We will prove that 
(1) $X_t$ contains each of the $n$ elements independently with probability $p$, and 
(2) given $X_t$, each element of $\Greedy(X_t)$ belongs to $A_t$ independently with probability $1-q(t)$. 
We adapt (2) from the observation of Abdi et al.~\cite[Lemma 5.3]{abdi2026strong} by replacing the probability $\frac{1}{2}$ with  $1-q(t)$.
Since this probability increases with time, we call (2) dynamic thinning.

Together, the two properties give the explicit formula of Proposition~\ref{prop:snapshot}. 
Moreover, at $t=1$, when all elements have arrived, an element of the fixed optimum belongs to $\Greedy(X_1)$ if and only if it belongs to $X_1$ (Lemma~\ref{lem:greedy-persistence}). 
By the two properties, such an element is selected with probability exactly $p\,(1-q(1))$. 
Choosing $p$ to maximize $p\,(1-q(1))$ gives the competitive ratio $C_*$.

\subsection{Related Work}
We discuss subsequent work and other related work on the matroid secretary problem and refer readers to the survey by Dinitz~\cite{dinitz2013recent} for earlier studies on the problem.

\paragraph{Subsequent Work.}
After the first version of this paper appeared, Huang~\cite{huang2026settling} settled the problem completely by giving an $e$-competitive algorithm (in fact, $e$-probability-competitive) with access to the same information as Singla's algorithm and superseded the result in this paper. 
Since the best competitive ratio for the classical secretary problem (i.e., at most one element may be selected) is $e$ \cite{babaioff2008online}, Huang's result also resolved the strong matroid secretary conjecture~\cite{soto2021strong} that there is an $e$-competitive algorithm for every matroid. 
Joyce~\cite{joyce2026multilevel} later extended Dynamic Thinning by maintaining several nested sets, obtaining an $(e+\varepsilon)$-competitive algorithm for every fixed $\varepsilon>0$ and bringing its competitive ratio arbitrarily close to the optimal value $e$.
In another work, Joyce~\cite{joyce2026refined} used the output distribution of Dynamic Thinning to obtain an $8.699$-competitive algorithm for the matroid secretary problem with a nonnegative monotone submodular objective.
The results of Singla~\cite{singla2026matroid}, Abdi et al.~\cite{abdi2026strong}, and Huang~\cite{huang2026settling} were subsequently combined and extended in Abdi et al.~\cite{abdi2026rankone}.

\paragraph{Other Related Work.}
Soto, Turkieltaub, and Verdugo~\cite{soto2021strong} introduced the notion of probability-competitive guarantees (i.e., per-element guarantees for online algorithms as discussed earlier). 
They studied ordinal online algorithms, which only need to know which of two elements has the larger weight, and gave online algorithms that are $4$-probability-competitive for graphic matroids and $3\sqrt{3}\approx5.19$-probability-competitive for laminar matroids. 
Bérczi et al.~\cite{berczi2025labeling} improved the guarantee for laminar matroids to $\frac{1}{1-\ln 2}\approx3.26$ via labeling schemes and gave a $3.71$-competitive online algorithm for graphic matroids of simple graphs.
Banihashem et al.~\cite{banihashem2025beating} gave a $3.95$-competitive online algorithm for graphic matroids of general graphs. 

Another line of work studies what information about the matroid an online algorithm needs in advance to achieve a constant competitive ratio. 
Santiago, Sergeev, and Zenklusen~\cite{santiago2025constant} showed that, in a variant of the problem where the weights are assigned to the elements at random, there is an online algorithm with a constant competitive ratio that knows only the number of elements. 
Before Singla's result, Dütting et al.~\cite{dutting2026graphic} gave a $36$-competitive online algorithm for the graphic matroid secretary problem when the graph is unknown to the algorithm.
Like Singla's algorithm, Huang's algorithm, and ours, it queries independence only on sets of already-arrived elements. 
On the negative side, Bahrani et al.~\cite{bahrani2022formal} showed that several natural online algorithms, including greedy-based and partition-based ones, cannot achieve a constant competitive ratio. 
Dughmi~\cite{dughmi2022matroid} further showed that the matroid secretary conjecture is equivalent to a problem called contention resolution, in which the goal is to select an independent subset of a randomly formed set of elements arriving in random order such that each element is selected with a guaranteed probability. 

\section{Preliminaries}
\label{sec:preliminaries}

A matroid is a pair $M=(E,\mathcal I)$, where $E$ is a finite ground set and $\mathcal I\subseteq 2^E$ is a nonempty family of independent sets.
The family $\mathcal I$ satisfies two properties. First, every subset of an independent set is independent: if $I\in\mathcal I$ and $I'\subseteq I$, then $I'\in\mathcal I$. Second, a smaller independent set can always be enlarged using an element of a larger one: if $I,J\in\mathcal I$ and $|I|<|J|$, then $I\cup\{e\}\in\mathcal I$ for some $e\in J\setminus I$. 

For $S\subseteq E$, its rank is $r(S)=\max\{|I|:I\subseteq S,\ I\in\mathcal I\}$, and its closure is $\cl(S)=\{e\in E:r(S\cup\{e\})=r(S)\}$.
We say that $S$ \emph{spans} $e$ if $e\in\cl(S)$, that is, adding $e$ does not increase the rank of $S$.
A set $C\subseteq E$ is \emph{closed} if $\cl(C)=C$.
A basis of $S$ is a maximal independent subset of $S$. 
Every basis $B$ of $S$ has size exactly $r(S)$, and it has the same closure as $S$, since every element of $S\setminus B$ is spanned by $B$. 
We also use the following standard facts about closure. First, closure is monotone: if $S\subseteq T$, then $\cl(S)\subseteq\cl(T)$. Second, if $C\subseteq D$, both $C$ and $D$ are closed, and $r(C)=r(D)$, then $C=D$. Third, for any $T\subseteq E$ and $g\in E$, $\cl(T\cup\{g\})=\cl(\cl(T)\cup\{g\})$. Moreover, $\cl(T)$ is itself closed.

In the matroid secretary problem, each element $e\in E$ has a fixed nonnegative weight $w_e$, initially unknown to the algorithm.
The elements arrive in a uniformly random order, and the identity and weight of an element are revealed upon its arrival.
The algorithm must immediately decide whether to accept (or select) or reject each arriving element, before observing the next one. 
The algorithm's decisions are irrevocable, and the accepted set must remain independent.
The algorithm knows $n=|E|$ in advance but need not know the ground set or the matroid structure.
It has access to an independence oracle, which answers whether a queried set belongs to $\mathcal I$, but may query only subsets of elements that have already arrived.
The matroid is also fixed in advance, and the algorithm's own random choices are independent of the instance and of the arrival order. 

We next specify how the greedy rule handles equal weights.  
Fix a total order on the element labels (i.e., a ranking of all elements), chosen independently of the arrival order. 
Whenever two elements have the same weight, the one earlier in this order is treated as having the larger weight. 
For $S\subseteq E$, let $\Greedy(S)$ denote the output of the greedy rule on $S$ with this fixed order. 
The greedy rule $\Greedy(S)$ scans the elements of $S$ with positive weight in decreasing order of weight, breaking ties by the fixed order, and includes an element whenever the set remains independent. 
This returns a maximum-weight independent subset of $S$.
We use $O=\Greedy(E)$ as the \emph{canonical optimum} or the fixed optimum. 
Since $O$ depends only on the weights, the matroid, and the fixed order, the optimum in a per-element guarantee does not depend on the algorithm's random choices. For $S\subseteq E$, we let $w(S)=\sum_{e\in S}w_e$.

Let $A$ denote the (random) accepted set returned by the algorithm. 
The algorithm is \emph{$c$-competitive} if $\E[w(A)]\ge\frac1c w(O)$ for every instance. 
The algorithm is $c$-probability-competitive~\cite{soto2021strong} if $\mathbb P[e\in A]\ge\frac1c$ for every $e\in O$. 
All probabilities and expectations are over both the random arrival order and the algorithm's own random choices. 
As stated in the introduction, the per-element guarantee implies the weight guarantee (i.e., a $c$-probability-competitive algorithm is $c$-competitive), because
\begin{align*}
\E[w(A)]
=\sum_{e\in E}w_e\Prb[e\in A]
\ge\sum_{e\in O}w_e\Prb[e\in A]
\ge\frac1c w(O).
\end{align*}
When proving a per-element guarantee, we assume that all weights are strictly positive. 
Indeed, zero-weight elements are not selected by the greedy rule and do not affect its output on any set. Hence they are not in $O$, and they are not accepted by our algorithm in Section~\ref{sec:algorithm}, which accepts only elements selected by the greedy rule. 

\paragraph{Basic Greedy Properties.}
We state two standard properties of the greedy rule on a matroid, which Singla~\cite{singla2026matroid} calls monotonicity and one-coordinate sensitivity.
The first property states that an element of the fixed optimum $O$ is selected by the greedy rule on any set that contains it. 
The second property states that adding an element that the greedy rule selects either enlarges the greedy output or replaces exactly one element of it, which we call the \emph{exchange partner} of the added element. We include the following proofs for completeness. 

\begin{lemma}
\label{lem:greedy-persistence}
For every $e\in O$ and $S\subseteq E\setminus\{e\}$, we have $e\in\Greedy(S\cup\{e\})$.
\end{lemma}

\begin{proof}
Let $P$ be the set of elements that the greedy rule scans before $e$ (i.e., the elements with larger weight than $e$, or with equal weight and earlier in the fixed order). Consider first the scan of $E$ that produces $O=\Greedy(E)$, and let $J$ be the set selected just before $e$ is scanned. Every element of $P\setminus J$ was rejected because it was spanned by the elements selected at that moment, which form a subset of $J$, so every element of $P$ is spanned by $J$. Since $J$ is an independent subset of $P$, it is a basis of $P$ and $\cl(J)=\cl(P)$. Because $e\in O$, the greedy rule selects $e$ at this point, which means $e\notin\cl(J)=\cl(P)$.

We now use this fact for the greedy rule scan of $S\cup\{e\}$. The elements scanned before $e$ are exactly those of $P\cap S$. 
By the same argument, the set $J'$ selected just before $e$ is a basis of $P\cap S$, so $\cl(J')=\cl(P\cap S)$.
Since closure is monotone, $\cl(P\cap S)\subseteq\cl(P)$, and therefore $e\notin\cl(J')$. Hence the greedy rule selects $e$. 
\end{proof}

\begin{lemma}
\label{lem:greedy-exchange}
Let $S\subseteq E\setminus\{e\}$ and suppose that $e\in\Greedy(S\cup\{e\})$.
If $e\notin\cl(S)$, then $\Greedy(S\cup\{e\})=\Greedy(S)\cup\{e\}$.
Otherwise, there is a unique $f\in\Greedy(S)$ such that $\Greedy(S\cup\{e\})=(\Greedy(S)\setminus\{f\})\cup\{e\}$.
\end{lemma}

\begin{proof}
Let us first compare the greedy rule scan of $S$ with the scan of $S\cup\{e\}$. 
Both scans process the elements of $S$ in the same order, but the second also processes $e$ at its position in that order. 
Before seeing $e$, the two scans make identical decisions. 
When the second scan reaches $e$, it selects $e$ by assumption. 
Up to the first element on which they decide differently, if there is one, their selected sets differ only in $e$. 
We write these sets as $J$ and $J\cup\{e\}$, where $J$ consists of the elements selected by both scans so far. 

We next show that, after $e$ has been scanned, the two scans make different decisions on at most one element of $S$. 
That is, there is at most one element of $S$ that one scan selects and the other skips. 
To show this, we use the sets $J$ and $J\cup\{e\}$ from above. 
Suppose the scans have made the same decisions so far and now scan an element $g$. 
The first scan selects $g$ if and only if $g\notin\cl(J)$, and the second selects $g$ if and only if $g\notin\cl(J\cup\{e\})$. 
Since $\cl(J)\subseteq\cl(J\cup\{e\})$, whenever the second scan selects $g$, so does the first. Hence the first element on which the scans make different decisions, if there is one, is an element $f$ that the first scan selects and the second skips. 

When $f$ is scanned, $J\cup\{f\}$ and $J\cup\{e\}$ are independent sets of size $|J|+1$, and $f\in\cl(J\cup\{e\})$. Thus $J\cup\{f\}\subseteq\cl(J\cup\{e\})$, which gives $\cl(J\cup\{f\})\subseteq\cl(J\cup\{e\})$. 
Both are closed sets of rank $|J|+1$, and the first is contained in the second, so they are equal by the second fact about closure in Section~\ref{sec:preliminaries}. 
Once the two selected sets have the same closure, both scans make the same decision on the next element, since they test membership in the same closure.
If both select an element $g$, their closures remain equal by the third fact about closure in Section~\ref{sec:preliminaries}.
Repeating this argument shows that the scans make the same decisions on all remaining elements.
Therefore the scans make different decisions on at most one element of $S$, and $\Greedy(S\cup\{e\})$ is either $\Greedy(S)\cup\{e\}$ or $(\Greedy(S)\setminus\{f\})\cup\{e\}$ for a single element $f$. 

It remains to determine whether the greedy output grows by one element or replaces an element.  
The output of each scan is a basis of its input set. 
Therefore, $|\Greedy(S)|=r(S)$ and $|\Greedy(S\cup\{e\})|=r(S\cup\{e\})$. 
If $e\notin\cl(S)$, then $r(S\cup\{e\})=r(S)+1$ by the definition of closure. Thus, the second scan skips no element selected by the first, and the first case occurs. 
If $e\in\cl(S)$, then $r(S\cup\{e\})=r(S)$ by the definition of closure. Thus, the second scan skips exactly one element selected by the first, and the second case occurs. 
Moreover, the skipped element $f$ is unique since it is the only element of $\Greedy(S)$ not in $\Greedy(S\cup\{e\})$.

\end{proof}

We will also use two immediate consequences of the proof of Lemma~\ref{lem:greedy-exchange}. 
First, if $e\notin\Greedy(S\cup\{e\})$, then the two greedy scans make the same decisions on all elements of $S$, so $\Greedy(S\cup\{e\})=\Greedy(S)$. 
Second, let $X$ be a reference set and $f\in\Greedy(X)$ an element of its reference basis. 
Applying Lemma~\ref{lem:greedy-exchange} with $S=X\setminus\{f\}$ and $e=f$ gives
\begin{align*}
\Greedy(X)\setminus\{f\}\subseteq\Greedy(X\setminus\{f\}).
\end{align*}
That is, deleting $f$ from the reference set keeps every other element of the reference basis in place, and at most one new element enters the reference basis to replace $f$. 
These two facts motivate the selection rule and the deletion rule of our algorithm in Section~\ref{sec:algorithm} and allow the algorithm to keep the accepted elements in the reference basis while unaccepted elements are replaced or deleted. 

\section{Algorithm and Analysis}
\label{sec:algorithm}

We call our algorithm \emph{Dynamic Thinning}.
Its parameter $p\in(0,1)$ controls both the expected fraction of elements used as the sample and the rate at which unaccepted elements of the reference basis are deleted. 
Let $C_*=3.14619322\ldots$ be the unique solution greater than $1$ of 
\begin{align}
C_*-\log C_*=2,
\qquad
p_*:=\frac{1}{C_*-1}.
\label{eq:optimal-constant}
\end{align}
The optimality of $p_*$ is established in Section~\ref{subsec:competitive-ratio}. 
Below, we state our main result. 

\begin{theorem}
\label{thm:main}
For every matroid secretary instance with $n$ elements and nonnegative weights, Dynamic Thinning with parameter $p_*$ always returns an independent set and accepts each element of the canonical optimum with probability exactly $\frac{1}{C_*}$.
Consequently, it is $C_*$-competitive, where $C_*\approx3.1462$.
The algorithm is ordinal and needs only the number of elements and independence queries on already-arrived elements.
\end{theorem}

\subsection{Algorithm and Implementation}
\label{subsec:algorithm}

Throughout the analysis, we assume that all weights are strictly positive, as justified in Section~\ref{sec:preliminaries}. 
The algorithm maintains a reference set $X$ of previously seen elements and the accepted set $A$, and it uses the reference basis $\Greedy(X)$ to decide whether an arriving element is a candidate for selection. 
Note that the reference basis is not the output of the algorithm, since it can contain sample elements that have already been rejected. 
Among its elements, only the accepted ones are protected. 
That is, the algorithm maintains $A\subseteq\Greedy(X)$ after every update (Lemma~\ref{lem:feasibility}), whereas an unaccepted element of the reference basis can be replaced by an accepted arriving element or deleted from $X$. A deletion changes only the reference set and it does not change any earlier decision to accept or reject an element.

To specify when deletions happen, the algorithm keeps track of time. 
It generates $n$ independent uniform random numbers in $(0,1)$, independently of the arrival order. 
The algorithm then assigns these numbers to the arriving elements in increasing order, so that the number assigned to element $e$ is its arrival time $\tau_e$ and later arriving elements have larger arrival times.
Given these numbers, the algorithm knows the time of the next arrival and can run the deletion rule on previously seen elements between consecutive arrivals. 

For $p\le t\le1$, let 
\begin{align}
q(t):=e^{-\frac{t-p}{p}}, 
\label{eq:thinning-schedule}
\end{align} 
where $q(p)=1$.
Here, $q(t)$ is the probability of accepting an arriving element that expands the reference basis. 
We will show that $1-q(t)$ is the probability that a given element of the reference basis has been accepted at time $t$, conditional on the reference set (Proposition~\ref{prop:snapshot}). 
We use $\mathrm{Exp}(\lambda)$ to denote the exponential distribution with rate $\lambda>0$: a waiting time drawn from this distribution exceeds $s\ge0$ with probability $e^{-\lambda s}$.
Its mean is $\frac{1}{\lambda}$, so a timer with mean $p$ has a waiting time drawn from $\mathrm{Exp}\big(\frac{1}{p}\big)$. 

\begin{algorithm}[Dynamic Thinning]
\label{alg:dynamic-thinning}
Given $p\in(0,1)$, the algorithm operates as follows.
\begin{enumerate}
\item \emph{Sample.} Generate the arrival times $(\tau_e)_{e\in E}$ as above and set
$K=|\{e\in E:\tau_e\le p\}|$.
Observe and reject the first $K$ arriving elements.  
Let $S$ denote this set. 
Set $X=S$, $A=\varnothing$, and $B=\Greedy(X)$, and initialize the current time to $p$. 

\item \emph{Deletion rule (between consecutive arrivals).}  
Let $s$ denote the current time and let $d=|B\setminus A|$ be the number of unaccepted elements of the reference basis.
If no arrivals remain, use time $1$ in place of the next arrival time in this step.
If $d=0$, advance to the next arrival time.
Otherwise, draw a waiting time $D\sim\mathrm{Exp}\big(\frac{d}{p}\big)$, whose mean is $\frac{p}{d}$.
If $s+D$ is earlier than the next arrival time, choose $f$ uniformly at random from $B\setminus A$, set 
\[
X\leftarrow X\setminus\{f\},\qquad
B\leftarrow\Greedy(X),\qquad s\leftarrow s+D,
\]
and repeat this step. Otherwise, advance to the next arrival time without any change. 

\item \emph{Selection rule (at an arrival).} 
When an element $e$ arrives at time $t=\tau_e>p$, compute $B^+=\Greedy(X\cup\{e\})$.
If $e\notin B^+$, reject $e$.
Otherwise, if $B^+=B\cup\{e\}$, accept $e$ with probability $q(t)$. 
Otherwise, by Lemma~\ref{lem:greedy-exchange}, there is a unique $f\in B$ such that $B^+=(B\setminus\{f\})\cup\{e\}$. 
Accept $e$ if $f\notin A$, and reject $e$ if $f\in A$. If $e$ is accepted, set 
\[
X\leftarrow X\cup\{e\},\qquad A\leftarrow A\cup\{e\},\qquad B\leftarrow B^+.
\]
If $e$ is rejected, leave $X$, $A$, and $B$ unchanged. Then return to Step~2.
\item \emph{Termination.} After the last arrival, run Step~2 until time $1$ and return $A$.
\end{enumerate}
\end{algorithm}

The selection rule ensures that, for a fixed arriving element $e$ at time $t>p$ and current sets $X,A$, the probability of acceptance is
\begin{align}
\Prb[\text{accept }e\mid X,A,t,e]
=
\begin{cases}
0,&e\notin B^+,\\
q(t),&B^+=B\cup\{e\},\\
0,&B^+=(B\setminus\{f\})\cup\{e\},\ f\in A,\\
1,&B^+=(B\setminus\{f\})\cup\{e\},\ f\notin A.
\end{cases}
\label{eq:acceptance-rule}
\end{align}


For the deletion, suppose $B\setminus A=\{f_1,\ldots,f_d\}$ and give each $f_i$ an independent timer $D_i\sim\mathrm{Exp}\big(\frac{1}{p}\big)$. 
The first timer runs out after time $D_{\min}=\min_i D_i$, which satisfies $\Pr[D_{\min}>s]=\prod_i\Pr[D_i>s]=e^{-ds/p}$. 
Hence $D_{\min}\sim\mathrm{Exp}\big(\frac{d}{p}\big)$, and since the timers are independent and identically distributed, each $f_i$ has probability $\frac{1}{d}$ of being deleted first. This gives the waiting time and the choice of the deleted element in Step~2. 
Step~2 also draws a new waiting time after every deletion and discards the current one at every arrival. By the memoryless property of exponential waiting times, this does not change the distribution of the process.

We first verify that the algorithm is well-defined and protects the accepted elements (i.e., every accepted element remains in the reference basis throughout all subsequent updates), as claimed in the introduction.

\begin{lemma}
\label{lem:feasibility}
After initialization and after every subsequent arrival (whether accepted or rejected) and every deletion of an element,
\begin{align}
A\subseteq\Greedy(X),\qquad A\cap S=\varnothing,
\qquad X\setminus A\subseteq S.
\label{eq:feasibility-invariant}
\end{align}
In particular, the accepted elements are protected, and Dynamic Thinning always returns an independent set.
\end{lemma}

\begin{proof}
We prove the three properties by induction over the arrivals and deletions of elements processed by the algorithm. 
Initially, $X=S$ and $A=\varnothing$, and all three properties hold. Suppose they hold before an arrival or a deletion of an element. 
We consider the three possible outcomes.

First, suppose an arriving element $e$ is rejected. Then neither $X$ nor $A$ changes, and the properties continue to hold. 

Second, suppose an arriving element $e$ is accepted. 
Then $e$ is added to both $X$ and $A$, and by Step~3 either $B^+=B\cup\{e\}$ or $B^+=(B\setminus\{f\})\cup\{e\}$ for some $f\notin A$. 
In both cases, the new accepted set $A\cup\{e\}$ is contained in the new reference basis $B^+$, since $A\subseteq B$ by the induction hypothesis and $f\notin A$ in the second case.
Thus, the first property holds. 
Since $e$ arrived after the sample, $e\notin S$, and the second property holds. 
The third holds because the only new element of $X$ is $e$, which is accepted.

Third, suppose an element $f$ is deleted. 
By the deletion rule, $f\in\Greedy(X)\setminus A$, and $A$ does not change. 
By Lemma~\ref{lem:greedy-exchange}, $\Greedy(X)\setminus\{f\}\subseteq\Greedy(X\setminus\{f\})$.
Thus, every element of $A$ remains in the reference basis and the first property holds. 
The second property still holds because $A$ does not change, and the third holds because no element is added to $X$. 
\end{proof}

The next lemma establishes the information requirements and the number of independence queries used by the algorithm. 

\begin{lemma}
\label{lem:online-implementation}
Dynamic Thinning uses only the number $n$ of elements and independence queries on already-arrived elements. It makes at most $n+1$ greedy scans and at most $n(n+1)$ independence queries. 
The algorithm is ordinal (i.e., its decisions depend on the weights only through their relative order, with ties broken by the fixed order). 
\end{lemma}

\begin{proof}
We first show that the algorithm uses only already-arrived elements. Initially, every element of $X=S$ has already arrived. A deletion removes a previously seen element, and an acceptance adds only the element that has just arrived.
Consequently, whenever the greedy rule is applied to $X$ or to $X\cup\{e\}$, every element involved has already arrived. Since the algorithm generates the arrival times in advance, it can carry out all deletions before $\tau_e$ without observing $e$.
It then decides whether to accept $e$ before observing the next arrival. 

We next count the greedy scans and independence queries. By Lemma~\ref{lem:feasibility}, every element that can be deleted belongs to $S$, and a deleted element is never added to $X$ again. Thus, there are at most $K$ deletions. 
The algorithm performs one greedy scan initially, one for each arrival after the sample, and one for each deletion, for a total of at most $1+(n-K)+K=n+1$ scans. 
Each scan makes at most one independence query per element scanned. 
Thus, a scan makes at most $n$ queries, and there are at most $n(n+1)$ queries in total.

All greedy scans use only the fixed relative order of the weights, and the remaining choices depend only on $p$, the arrival times, and the maintained sets.
Thus, the algorithm is ordinal.
\end{proof}

For the computation time, we count sampling a uniform or exponential random number, evaluating $q(t)$, and comparing arrival times as single operations. 
Under this convention, generating and sorting the arrival times and performing all greedy scans and updates take polynomial time. 

\subsection{Arrival Times and States}
\label{subsec:snapshot}

We first describe the distribution of the sample $S$ determined by the arrival times.
The algorithm generates $n$ independent uniform numbers and assigns them, in increasing order, to the elements in their arrival order.
Because the arrival order is uniformly random and independent of these numbers, the resulting arrival times $(\tau_e)_{e\in E}$ are independent and uniform on $(0,1)$.
This gives the following observation.

\begin{observation}
\label{lem:clocks-and-order}
Suppose that the arrival times $(\tau_e)_{e\in E}$ are independent uniform random variables on $(0,1)$.
Then, $K\sim\operatorname{Bin}(n,p)$ and $S=\{e:\tau_e\le p\}$ is an independent Bernoulli-$p$ sample of $E$ (i.e., $S$ contains each element of $E$ independently with probability $p$).  
Therefore, for any fixed $s\subseteq E$ of size $k$,
\begin{equation*}
\Prb[S=s]
=\frac{\Prb[K=k]}{\binom nk}
=p^k(1-p)^{n-k}.\qedhere
\end{equation*}
\end{observation}


For a fixed $t\in[p,1]$, let $X_t$ and $A_t$ denote the reference set and the accepted set after all arrivals and deletions up to time $t$ have been processed. 
In addition, define 
\begin{align*}
U_t:=\{e:\tau_e>t\},
\qquad
Z_t:=E\setminus(X_t\cup U_t).
\end{align*}
That is, $U_t$ is the set of elements that have not yet arrived at time $t$, and $Z_t$ is the set of elements that have arrived but are outside the reference set.
The latter consists of the rejected elements and the deleted sample elements. 
Since every element of $X_t$ has arrived, $(X_t,Z_t,U_t)$ is a partition of $E$.

Given these sets and the time $t$, the probabilities of subsequent arrivals, deletions, and selection decisions do not depend on the earlier states.
To see this, consider the history of arrivals, deletions, and selection decisions up to time $t<1$.
This history does not include the arrival times or waiting times of future events. 
After time $t$, the algorithm either processes an arriving element or deletes an element whose timer has run out. 

Let us first consider the arrival of an element. 
For each element in $U_t$, this history reveals only that its arrival time exceeds $t$.
The arrival times of these elements therefore remain independent and uniform on $(t,1)$. Hence, for each $e\in U_t$, the probability that $e$ arrives in $(t,t+h]$ is $\frac{h}{1-t}$ for $0<h<1-t$. In other words, $e$ arrives at rate $\frac{1}{1-t}$. 

Next, consider the deletion of an element. 
Each element $f\in\Greedy(X_t)\setminus A_t$ has its own timer drawn from $\mathrm{Exp}\big(\frac{1}{p}\big)$. 
By the memoryless property of the exponential distribution, the probability that the timer runs out in $(t,t+h]$ is $1-e^{-h/p}$, which is approximately $\frac{h}{p}$ for small $h$, no matter how long the timer has already been running. 
Thus, $f$ is deleted at rate $\frac{1}{p}$. 
Finally, the (acceptance or rejection) decision for an arriving element is made by the selection rule~\eqref{eq:acceptance-rule}, which uses only $X_t$, $A_t$, the arriving element, and the time $t$.

We define the state at time $t$ as
\[
H_t:=(X_t,Z_t,U_t,A_t).
\]
The probabilities of the next arrival, deletion, and selection decision depend only on $H_t$ and the time $t$, through the arrival rate $\frac{1}{1-t}$ and the acceptance probability $q(t)$.
Thus, $(H_t)_{t\in[p,1]}$ is a finite-state Markov process, with transition rates that depend on $t$ (see Norris~\cite{norris1997markov} for background on continuous-time Markov chains). 
The state probabilities satisfy a system of differential equations determined by these rates.
We use the uniqueness of its solution for a given initial distribution to prove Proposition~\ref{prop:snapshot}. 

\subsection{Analysis of Dynamic Thinning}
\label{subsec:dynamic-thinning-analysis}

We first define thinning, as used in the introduction. 
An \emph{independent $\alpha$-thinning} of a set $B$ is the random subset of $B$ obtained by retaining each element of $B$ independently with probability $\alpha$. 
Property (2) in the introduction states that, given $X_t$, the accepted set $A_t$ is an independent $(1-q(t))$-thinning of the reference basis $\Greedy(X_t)$.
We adapted this property from the conditional-thinning lemma of Abdi et al.~\cite[Lemma 5.3]{abdi2026strong}. 

In their process, the reference set is updated only at arrivals, each arrival may change only the status of the arriving element, and the thinning probability is fixed at $\frac{1}{2}$. 
In our algorithm, elements of the reference set may be deleted between arrivals, each element belongs to the reference set with probability $p$, and the thinning probability $1-q(t)$ depends on time. 
In the formula below, we use the convention $0^0=1$ at $t=p$ and $t=1$. 

\begin{proposition}[Joint distribution]
\label{prop:snapshot}
For every $t\in[p,1]$, every partition $(x,z,u)$ of $E$ (i.e., pairwise disjoint sets with $x\cup z\cup u=E$), and every $a\subseteq\Greedy(x)$, 
\begin{align}
&\Prb[X_t=x,Z_t=z,U_t=u,A_t=a] =p^{|x|}(t-p)^{|z|}(1-t)^{|u|}
(1-q(t))^{|a|}q(t)^{r(x)-|a|}.
\label{eq:joint-snapshot}
\end{align}
States with $a\not\subseteq\Greedy(x)$ have probability zero. 
Consequently, each element of $E$ belongs to $X_t$, $Z_t$, or $U_t$ independently with probabilities $p$, $t-p$, and $1-t$, respectively. 
Moreover, given the partition $(X_t,Z_t,U_t)$, the accepted set $A_t$ is an independent $(1-q(t))$-thinning of $\Greedy(X_t)$.
\end{proposition}

\paragraph{Differential Equations.}
We call a tuple $(x,z,u,a)$ a \emph{valid state} if $(x,z,u)$ is a partition of $E$ and $a\subseteq\Greedy(x)$.
Let $P_t(\xi)=\Prb[H_t=\xi]$ be the probability of state $\xi=(x,z,u,a)$ at time $t$.
Let $\mu_t(\xi)$ denote the right-hand side of~\eqref{eq:joint-snapshot} for a valid state, and set $\mu_t(\xi)=0$ otherwise.
We first show that $P_t$ and $\mu_t$ satisfy the same differential equations.
We then check that they have the same initial distribution and use uniqueness of the solution to prove Proposition~\ref{prop:snapshot}.

For distinct valid states $\xi$ and $\eta$, let $r_t(\xi,\eta)$ be the rate at which a single arrival or deletion changes the state from $\xi$ to $\eta$.
These rates are given by the arrival and deletion rates in Section~\ref{subsec:snapshot}, together with the selection probabilities in~\eqref{eq:acceptance-rule}.
The rate is zero if no single arrival or deletion can make this change.

\begin{lemma}
\label{lem:state-equations}
For every valid state $\xi$ and every $p<t<1$,
\begin{equation}
\label{eq:forward}
P_t'(\xi)=\sum_{\eta\ne\xi}P_t(\eta)r_t(\eta,\xi)-P_t(\xi)\sum_{\eta\ne\xi}r_t(\xi,\eta),
\end{equation}
where both sums range over valid states.
The equations also hold at $t=p$ with right derivatives, and their coefficients are continuous on every interval $[p,T]$ with $T<1$.
\end{lemma}

\begin{proof}
Fix a valid state $\xi$ and a time $p\le t<1$.
For a small $h>0$, let $I_h=\Prb[H_t\ne\xi,\ H_{t+h}=\xi]$ and $O_h=\Prb[H_t=\xi,\ H_{t+h}\ne\xi]$.
These are the probabilities of entering and leaving $\xi$ over the interval $(t,t+h]$, respectively.
Then
\[
P_{t+h}(\xi)-P_t(\xi)=I_h-O_h.
\]
Dividing by $h$ and taking the limit gives the derivative of $P_t(\xi)$ as the inflow rate minus the outflow rate.

To compute these rates, it suffices to consider a single arrival or deletion.
Indeed, fix $T$ with $t<T<1$.
On $[p,T]$, there are at most $n$ elements that can arrive, each at rate at most $\frac{1}{1-T}$, and at most $n$ elements that can be deleted, each at rate $\frac1p$.
Thus, the total rate of leaving any state is at most $\frac{n}{1-T}+\frac{n}{p}$.
The probability of two or more changes in an interval of length $h$ is therefore bounded by a constant times $h^2$, uniformly over these states and times, and vanishes after division by $h$ as $h\to0^+$.
It follows that
\[
\lim_{h\to0^+}\frac{I_h}{h}=\sum_{\eta\ne\xi}P_t(\eta)r_t(\eta,\xi),
\qquad
\lim_{h\to0^+}\frac{O_h}{h}=P_t(\xi)\sum_{\eta\ne\xi}r_t(\xi,\eta).
\]
Subtracting gives~\eqref{eq:forward}.

Each transition rate is a sum of terms formed from the arrival rate $\frac{1}{1-t}$, the deletion rate $\frac1p$, and the selection probabilities $q(t)$, $1-q(t)$, $0$, or $1$.
These quantities are continuous on $[p,T]$, which also gives continuity of the coefficients and the right-derivative statement at $t=p$.
\end{proof}

\paragraph{Inflow and Outflow.}
We now substitute $\mu_t$ into the equations of Lemma~\ref{lem:state-equations}.
We consider outflow, inflow through acceptances, and inflow through rejections and deletions separately.

\begin{lemma}
\label{lem:formula-equations}
The functions $\mu_t(\xi)$ satisfy~\eqref{eq:forward} with $\mu_t$ in place of $P_t$ for every valid state $\xi$ and every $p<t<1$.
\end{lemma}

\begin{proof}
Fix a valid state $\xi=(x,z,u,a)$ and a time $t$ with $p<t<1$. 
For simplicity, let $q=q(t)$ and $\mu=\mu_t(\xi)$.
Since $p<t<1$, all factors defining $\mu$ are positive, so $\mu>0$. 
We first differentiate~\eqref{eq:joint-snapshot} with respect to $t$ and divide by $\mu$. 
Using $q'=-\frac{q}{p}$, we have 
\begin{align}
\frac{1}{\mu}\cdot\frac{d\mu}{dt}=\frac{|z|}{t-p}-\frac{|u|}{1-t}+\frac{|a|q}{p(1-q)}-\frac{r(x)-|a|}{p}.
\label{eq:derivative}
\end{align}
We compute the inflow and outflow rates in~\eqref{eq:forward}, using $\mu_t$ in place of $P_t$.
We will show that their difference, divided by $\mu$, equals the right-hand side of~\eqref{eq:derivative}.

\emph{Outflow.} The outflow rate is $\mu$ times the total rate at which the process leaves $\xi$.
Each $e\in u$ arrives at rate $\frac{1}{1-t}$, and its arrival changes the state.
If $e$ is accepted, it moves from $u$ into both $x$ and $a$, and if it is rejected, it moves from $u$ into $z$. 
The total outflow rate due to arrivals is then $\frac{\mu|u|}{1-t}$.
Each $f\in\Greedy(x)\setminus a$ is deleted at rate $\frac{1}{p}$, and its deletion changes the state to $(x\setminus\{f\},z\cup\{f\},u,a)$. 
Because there are $r(x)-|a|$ such elements, the total outflow rate due to deletions is $\frac{\mu(r(x)-|a|)}{p}$. 
Dividing the total outflow rate by $\mu$ gives the two negative terms of~\eqref{eq:derivative}.

\emph{Inflow through acceptances.} 
An acceptance that leads to $\xi$ must add some element $e\in a$.
We consider each such element in turn. 
The only state from which accepting $e$ leads to $\xi$ is $\xi'=(x\setminus\{e\},z,u\cup\{e\},a\setminus\{e\})$. 
This state is valid: since $e\in\Greedy(x)$, Lemma~\ref{lem:greedy-exchange} gives $\Greedy(x)\setminus\{e\}\subseteq\Greedy(x\setminus\{e\})$, and hence $a\setminus\{e\}\subseteq\Greedy(x\setminus\{e\})$.
To compute the inflow rate from $\xi'$, we need $\mu_t(\xi')$ and the probability that $e$ is accepted when it arrives in state $\xi'$.
Both depend on whether removing $e$ from $x$ lowers the rank, so we consider two cases. 

If $r(x\setminus\{e\})=r(x)-1$, then substituting $\xi'$ into~\eqref{eq:joint-snapshot} gives $\mu_t(\xi')=\mu\cdot\frac{1-t}{p(1-q)}$.
In this case $\Greedy(x)=\Greedy(x\setminus\{e\})\cup\{e\}$ by Lemma~\ref{lem:greedy-exchange}, which means that when $e$ arrives in state $\xi'$, it expands the reference basis and is accepted with probability $q$.
Since $e$ arrives at rate $\frac{1}{1-t}$, the inflow rate from $\xi'$ is $\mu\cdot\frac{1-t}{p(1-q)}\cdot\frac{1}{1-t}\cdot q=\mu\,\frac{q}{p(1-q)}$. 

If $r(x\setminus\{e\})=r(x)$, then substituting $\xi'$ into~\eqref{eq:joint-snapshot} gives $\mu_t(\xi')=\mu\cdot\frac{(1-t)q}{p(1-q)}$.
In this case Lemma~\ref{lem:greedy-exchange} gives a unique $f\notin\Greedy(x)$ with $\Greedy(x\setminus\{e\})=(\Greedy(x)\setminus\{e\})\cup\{f\}$, which means that when $e$ arrives in state $\xi'$, it replaces $f$ in the reference basis.
Since $a\subseteq\Greedy(x)$ and $f\notin\Greedy(x)$, we have $f\notin a\setminus\{e\}$, and the selection rule accepts $e$ with probability one.
The inflow rate from $\xi'$ is $\mu\cdot\frac{(1-t)q}{p(1-q)}\cdot\frac{1}{1-t}=\mu\,\frac{q}{p(1-q)}$.

In both cases, each $e\in a$ contributes $\frac{\mu q}{p(1-q)}$ to the inflow rate.
Summing over $e\in a$, the total inflow rate through acceptances is $\frac{\mu|a|q}{p(1-q)}$.
Dividing by $\mu$ gives the third term of~\eqref{eq:derivative}.

\emph{Inflow through rejections and deletions.} 
A rejection or deletion that leads to $\xi$ adds an element $e$ to $z$ and leaves $a$ unchanged.
For each $e\in z$, there are two possible preceding tuples: one in which $e$ has not yet arrived, and one in which $e$ is an unaccepted element of the reference set.
We compute their contributions separately.

The first state is $\xi_r=(x,z\setminus\{e\},u\cup\{e\},a)$, in which $e$ has not yet arrived.
Substituting into~\eqref{eq:joint-snapshot} gives $\mu_t(\xi_r)=\mu\cdot\frac{1-t}{t-p}$.
Let $\beta_e$ denote the probability that the selection rule rejects $e$ when the reference set is $x$ and the accepted set is $a$.
The inflow rate from $\xi_r$ is $\mu\cdot\frac{1-t}{t-p}\cdot\frac{1}{1-t}\cdot\beta_e=\frac{\mu\beta_e}{t-p}$.

The second state is $\xi_d=(x\cup\{e\},z\setminus\{e\},u,a)$, in which $e$ is an unaccepted element of the reference set.
Let $B=\Greedy(x)$ and $B^+=\Greedy(x\cup\{e\})$.
The state $\xi_d$ is valid only if $a\subseteq B^+$, and $e$ can be deleted from it only if $e\in B^+\setminus a$.
When both conditions hold, $e$ is deleted at rate $\frac{1}{p}$, and the inflow rate from $\xi_d$ is $\frac{\mu_t(\xi_d)}{p}$.
To determine $\beta_e$ and the inflow rate from $\xi_d$, we consider the four cases of the selection rule~\eqref{eq:acceptance-rule} for reference set $x$ and accepted set $a$. 
These cases are (i) $e\notin B^+$, (ii) $B^+=B\cup\{e\}$, (iii) $B^+=(B\setminus\{f\})\cup\{e\}$ with $f\in a$, and (iv) $B^+=(B\setminus\{f\})\cup\{e\}$ with $f\notin a$.

\emph{Case (i): $e\notin B^+$.} The selection rule rejects $e$, hence $\beta_e=1$ and the inflow rate from $\xi_r$ is $\frac{\mu}{t-p}$.
In $\xi_d$, the element $e$ is not in the reference basis $B^+$ and has no timer.
Thus $\xi_d$ does not contribute to the inflow rate, and the total is $\frac{\mu}{t-p}$.

\emph{Case (ii): $B^+=B\cup\{e\}$.} The arrival of $e$ expands the reference basis and is accepted with probability $q$, hence $\beta_e=1-q$ and the inflow rate from $\xi_r$ is $\frac{\mu(1-q)}{t-p}$.
The state $\xi_d$ is valid because $a\subseteq B\subseteq B^+$, and $e\in B^+\setminus a$ because $e\in z$ in $\xi$.
Substituting into~\eqref{eq:joint-snapshot} gives $\mu_t(\xi_d)=\mu\cdot\frac{pq}{t-p}$, because moving $e$ from $z$ into the reference set contributes the factor $\frac{p}{t-p}$, and the rank increases by one with $a$ unchanged, which contributes the factor $q$.
The inflow rate from $\xi_d$ is $\mu\cdot\frac{pq}{t-p}\cdot\frac1p=\frac{\mu q}{t-p}$.
The total is $\frac{\mu(1-q)}{t-p}+\frac{\mu q}{t-p}=\frac{\mu}{t-p}$.

\emph{Case (iii): $B^+=(B\setminus\{f\})\cup\{e\}$ with $f\in a$.} Accepting $e$ would replace an accepted element, and the selection rule rejects $e$.
Hence $\beta_e=1$ and the inflow rate from $\xi_r$ is $\frac{\mu}{t-p}$.
The state $\xi_d$ is not valid, since its accepted set $a$ contains $f$ but its reference basis $B^+$ does not.
Thus $\mu_t(\xi_d)=0$, and $\xi_d$ does not contribute.
The total is again $\frac{\mu}{t-p}$.

\emph{Case (iv): $B^+=(B\setminus\{f\})\cup\{e\}$ with $f\notin a$.} The selection rule accepts $e$ with probability one.
Hence $\beta_e=0$, and $\xi_r$ does not contribute.
The state $\xi_d$ is valid because $a\subseteq B\setminus\{f\}\subseteq B^+$, and $e\in B^+\setminus a$ has a timer.
The rank and the size of the accepted set in $\xi_d$ equal those in $\xi$, and substituting into~\eqref{eq:joint-snapshot} gives $\mu_t(\xi_d)=\mu\cdot\frac{p}{t-p}$ with no additional factor of $q$.
The inflow rate from $\xi_d$ is $\mu\cdot\frac{p}{t-p}\cdot\frac1p=\frac{\mu}{t-p}$, which is also the total in this case.

The four cases are summarized in the following table, where each entry is the inflow rate divided by $\frac{\mu}{t-p}$.
\begin{center}
\begin{tabular}{lccc}
\hline
Case & Rejection of $e$ & Deletion of $e$ & Total \\ \hline
(i) $e\notin B^+$ & $1$ & $0$ & $1$\\
(ii) $B^+=B\cup\{e\}$ & $1-q$ & $q$ & $1$\\
(iii) $B^+=(B\setminus\{f\})\cup\{e\}$, $f\in a$ & $1$ & $0$ & $1$\\
(iv) $B^+=(B\setminus\{f\})\cup\{e\}$, $f\notin a$ & $0$ & $1$ & $1$\\ \hline
\end{tabular}
\end{center}

In every case, the inflow rates associated with $e$ add up to $\frac{\mu}{t-p}$.
Summing over $e\in z$, the total inflow rate through rejections and deletions is $\frac{\mu|z|}{t-p}$.
Dividing by $\mu$ gives the first term of~\eqref{eq:derivative}.
Every state change is an acceptance, a rejection, or a deletion, so these calculations account for all inflow and outflow of $\xi$.
Their difference, divided by $\mu$, equals the right-hand side of~\eqref{eq:derivative}.
Thus, $\mu_t$ satisfies~\eqref{eq:forward} with $\mu_t$ in place of $P_t$ for every valid state and every $p<t<1$.
\end{proof}

\paragraph{Joint Distribution.}
The two preceding lemmas give the differential equations for both $P_t$ and $\mu_t$.
We now check the initial distribution and use uniqueness to prove Proposition~\ref{prop:snapshot}.

\begin{proof}[Proof of Proposition~\ref{prop:snapshot}]
Fix a partition $(x,z,u)$. 
Recall that $\Greedy(x)$ is a basis of $x$ and that all weights are positive. 
Therefore, we have $|\Greedy(x)|=r(x)$.
Summing the last two factors of~\eqref{eq:joint-snapshot} over all $a\subseteq\Greedy(x)$ and using the binomial theorem gives
$$\sum_{a\subseteq\Greedy(x)}(1-q(t))^{|a|}q(t)^{r(x)-|a|}=\sum_{k=0}^{r(x)}\binom{r(x)}{k}(1-q(t))^{k}q(t)^{r(x)-k}=1.$$
Hence, the total probability that $\mu_t$ assigns to the valid states with partition $(x,z,u)$ is $p^{|x|}(t-p)^{|z|}(1-t)^{|u|}$. 

Summing this over all partitions $(x,z,u)$ of $E$ gives
$$\sum_{(x,z,u)}p^{|x|}(t-p)^{|z|}(1-t)^{|u|}=\sum_{i+j+k=n}\frac{n!}{i!\,j!\,k!}\,p^{i}(t-p)^{j}(1-t)^{k}=\big(p+(t-p)+(1-t)\big)^n=1,$$
where the first equality groups the partitions by the sizes $i=|x|$, $j=|z|$, and $k=|u|$, and the second equality is due to the multinomial theorem. 
The values of $\mu_t$ are nonnegative and add up to $1$. 
Hence, $\mu_t$ is a probability distribution over valid states.

At time $t=p$, the reference set $X_p$ is the sample $S$, the sets $A_p$ and $Z_p$ are empty, and $U_p=E\setminus S$.
Since $q(p)=1$, formula~\eqref{eq:joint-snapshot} at $t=p$ assigns probability $p^{|x|}(1-p)^{n-|x|}$ to the state $(x,\varnothing,E\setminus x,\varnothing)$ and probability zero to every other state. 
By Observation~\ref{lem:clocks-and-order}, $\Prb[S=x]=p^{|x|}(1-p)^{n-|x|}$. 
Hence, $\mu_p=P_p$.

Fix $T$ with $p<T<1$.
By Lemma~\ref{lem:state-equations}, the functions $P_t(\xi)$ satisfy a finite system of linear differential equations with continuous coefficients on $[p,T]$.
Such a system has a unique solution for a given initial condition~\cite{teschl2012ordinary}.
Lemma~\ref{lem:formula-equations} shows that $\mu_t$ satisfies the same equations for $p<t\le T$.
The products defining $\mu_t$ and their derivatives are continuous at $t=p$, so the equations also hold there by taking the limit as $t$ decreases to $p$.
Since $\mu_p=P_p$, uniqueness gives $\mu_t=P_t$ on $[p,T]$.
As $T<1$ was arbitrary, this proves~\eqref{eq:joint-snapshot} for $p\le t<1$.

It remains to consider $t=1$.
There are finitely many arrivals and deletions, and with probability one none occurs exactly at time $1$.
Hence $H_t=H_1$ for all $t$ sufficiently close to $1$ on each such realization.
Taking the limit of the state probabilities and of the formula as $t$ increases to $1$ proves~\eqref{eq:joint-snapshot} at $t=1$.

We now show the two remaining claims of the lemma.
Summing~\eqref{eq:joint-snapshot} over all $a\subseteq B$, where $B=\Greedy(x)$, gives
$$\Prb[X_t=x,\,Z_t=z,\,U_t=u]=p^{|x|}(t-p)^{|z|}(1-t)^{|u|}.$$
This assigns a factor $p$ to each element of $x$, a factor $t-p$ to each element of $z$, and a factor $1-t$ to each element of $u$, which is exactly the distribution obtained by placing each element independently into $X_t$, $Z_t$, or $U_t$ with probabilities $p$, $t-p$, and $1-t$.
This proves the first remaining claim.

Next, fix a partition $(x,z,u)$ with positive probability.
Dividing~\eqref{eq:joint-snapshot} by the display above gives
$$\Prb[A_t=a\mid X_t=x,\,Z_t=z,\,U_t=u]=(1-q(t))^{|a|}q(t)^{|B|-|a|}$$
for every $a\subseteq B$.
This is the probability of obtaining exactly $a$ when each element of $B$ is retained independently with probability $1-q(t)$.
Hence, given the partition, $A_t$ is an independent $(1-q(t))$-thinning of $\Greedy(X_t)$, which proves the second remaining claim.
\end{proof}


\subsection{Competitive Ratio}
\label{subsec:competitive-ratio}

We now apply Proposition~\ref{prop:snapshot} at time $t=1$ to compute the probability that an element of the canonical optimum $O$ is accepted.

\begin{proposition}
\label{prop:selection-probability}
For every $p\in(0,1)$ and $e\in O$, Dynamic Thinning with parameter $p$ satisfies 
\begin{align}
\Prb[e\in A]
=\gamma(p),
\qquad
\gamma(p):=p\left(1-e^{-\frac{1-p}{p}}\right).
\label{eq:exact-selection}
\end{align}
More generally, if $R_p\subseteq E$ is a random set that contains each element of $E$ independently with probability $p$, then
\begin{align}
\E[w(A)]
=\left(1-e^{-\frac{1-p}{p}}\right)
\E\bigl[w(\Greedy(R_p))\bigr].
\label{eq:exact-reward}
\end{align}
In particular, the algorithm is $\frac{1}{\gamma(p)}$-probability-competitive.
\end{proposition}

\begin{proof}

At time $1$, every element has arrived. 
Thus, $U_1=\varnothing$ and $Z_1=E\setminus X_1$.
Proposition~\ref{prop:snapshot} at $t=1$ gives, for every $x\subseteq E$ and $a\subseteq\Greedy(x)$,
\begin{align}
\Prb[X_1=x,A_1=a]=p^{|x|}(1-p)^{n-|x|}(1-q(1))^{|a|}q(1)^{r(x)-|a|}.
\label{eq:terminal}
\end{align}
Hence, each element of $E$ belongs to $X_1$ independently with probability $p$, and given $X_1$, each element of $\Greedy(X_1)$ belongs to $A_1$ independently with probability $1-q(1)$.
Since the algorithm returns $A=A_1$, linearity of expectation gives $\E[w(A)\mid X_1]=(1-q(1))\,w(\Greedy(X_1))$.
As $X_1$ and $R_p$ have the same distribution, we have 

$$\E[w(A)]=\E\big[\E[w(A)\mid X_1]\big]=(1-q(1))\,\E\big[w(\Greedy(R_p))\big].$$

This is exactly~\eqref{eq:exact-reward}, since $q(1)=e^{-(1-p)/p}$.

Now fix $e\in O$.
By Lemma~\ref{lem:greedy-persistence}, $e\in\Greedy(x)$ for every $x\subseteq E$ containing $e$, so $e\in\Greedy(X_1)$ if and only if $e\in X_1$.
Conditional on $X_1$, an element of $\Greedy(X_1)$ is accepted with probability $1-q(1)$.
Since $\Prb[e\in X_1]=p$, we obtain
\begin{align*}
\Prb[e\in A]=(1-q(1))\Prb[e\in X_1]=p\left(1-e^{-\frac{1-p}{p}}\right),
\end{align*}
which is~\eqref{eq:exact-selection}. 
The last claim follows from the definition of probability-competitiveness in Section~\ref{sec:preliminaries}.

\end{proof}

\begin{proof}[Proof of Theorem~\ref{thm:main}]
Lemmas~\ref{lem:feasibility} and~\ref{lem:online-implementation} establish feasibility, ordinality, and the information requirements.
It remains to choose $p$ to maximize $\gamma(p)$ in~\eqref{eq:exact-selection}.
For $0<p<1$, differentiation gives

$$
\gamma'(p)=1-\left(1+\frac1p\right)e^{1-\frac1p},\qquad
\gamma''(p)=-\frac{1}{p^3}e^{1-\frac1p}<0.
$$

Moreover, $\gamma'(p)\to1$ as $p$ decreases to $0$ and $\gamma'(p)\to-1$ as $p$ increases to $1$.
Hence $\gamma$ is strictly concave on $(0,1)$ and has a unique maximizer $p_*\in(0,1)$, determined by $\gamma'(p_*)=0$. 
Equivalently, 
$$e^{1-\frac{1}{p_*}}=\frac{p_*}{1+p_*}.$$

At $p=p_*$, we have $\gamma(p_*)=\frac{p_*}{1+p_*}$. 
Writing $C_*=\frac{1}{\gamma(p_*)}=1+\frac{1}{p_*}$, the equation $e^{1-1/p_*}=\frac{p_*}{1+p_*}$ becomes $e^{2-C_*}=\frac{1}{C_*}$ (i.e., $C_*-\log C_*=2$). 
This proves~\eqref{eq:optimal-constant} and yields

$$p_*\approx0.4659412724,\qquad\gamma(p_*)\approx0.3178444329,\qquad C_*\approx3.1461932206.$$

Thus, with parameter $p_*$, the algorithm accepts each element of $O$ with probability exactly $\frac{1}{C_*}$ by Proposition~\ref{prop:selection-probability}, and it is $C_*$-competitive by the definition in Section~\ref{sec:preliminaries}. 
\end{proof}

Finally, the guarantee is tight for each fixed $p$.
Consider the free matroid, in which every subset of $E$ is independent.
For positive weights, $O=E$ and, by Proposition~\ref{prop:selection-probability}, every element is accepted with probability exactly $\gamma(p)$.
Hence $\frac{1}{\gamma(p)}$ is the exact competitive ratio of Dynamic Thinning with parameter $p$ on this instance, and $p_*$ is the best choice of parameter within this family of algorithms.
For simple rational choices of the parameter, $p=\frac12$ gives a competitive ratio of approximately $3.1640$, while $p=\frac{27}{58}$ gives approximately $3.1462$, close to the optimum.

\section{Conclusion}
\label{sec:conclusion}

We gave a $3.1462$-competitive online algorithm for the matroid secretary problem that uses only the number of elements and independence queries on already-arrived elements.
The algorithm protects only the accepted elements and allows unaccepted elements of the reference basis to be deleted or replaced.
It combines random deletions at a fixed rate with acceptance probabilities that depend on the arrival time.
These rules allow us to derive the joint distribution of the reference set and the accepted set at every time, from which we obtain the acceptance probability of each element of the optimum and the competitive ratio.

After the first version of this paper appeared, Huang~\cite{huang2026settling} gave an $e$-competitive algorithm with access to the same information, achieving the best possible competitive ratio.
Joyce~\cite{joyce2026multilevel} later extended Dynamic Thinning by maintaining several nested sets, obtaining an $(e+\varepsilon)$-competitive algorithm for every fixed $\varepsilon>0$ and bringing its competitive ratio arbitrarily close to the optimal value $e$.

\paragraph{Use of AI Tools.}
While working on the graphic matroid secretary problem, we learned of Singla's work~\cite{singla2026matroid} resolving the general matroid secretary conjecture and studied his approach.
We discussed with GPT-6 Astra whether this approach could be combined with our existing results for graphic matroids.
During these discussions, Astra proposed a sample-only reserve construction for the general problem, yielding the simple $3.7321$-competitive algorithm presented in the first version of this paper.

We later considered Abdi et al.'s conditional-thinning method~\cite{abdi2026strong} and explored its application to the sample-only reserve construction through further discussions with Astra.
These discussions led to the dynamic reference process and the $3.1462$-competitive algorithm presented in this version.
The principal algorithmic constructions and proof arguments were generated by Astra through these discussions.
Codex was used to assist with revising the manuscript.
We carefully checked the algorithms and proofs, revised the arguments before releasing each version, and take responsibility for the mathematical claims and the final content.

\paragraph{Revisions in this Version.}
We revised the paper to improve readability, clarify definitions, and keep terminology consistent across sections.
The algorithm and mathematical results are unchanged, but we reorganized parts of the proofs and added explanations.

During the rewriting, Claude (Anthropic) helped check whether revised passages preserved the original meaning, explain technical steps in simpler language, and suggest alternative phrasings.
We reviewed the revised text and chose the final wording.

\bibliographystyle{plain}
\bibliography{improved-thinning-embedded}

\end{document}